\documentclass[conference]{IEEEtran}
\IEEEoverridecommandlockouts
\usepackage{cite,url}
\usepackage{amsmath,amssymb,amsfonts,physics}
\usepackage{amsthm}
\newtheorem{theorem}{Theorem}
\usepackage{algorithmic}
\usepackage{graphicx}
\usepackage{textcomp}
\usepackage{xcolor}
\usepackage{pgfplots}
\usepackage{float}
\def\BibTeX{{\rm B\kern-.05em{\sc i\kern-.025em b}\kern-.08em
    T\kern-.1667em\lower.7ex\hbox{E}\kern-.125emX}}

\newcommand{\edit}[1]{{#1}}

\begin{document}

\title{An Economic Model for Quantum Key-Recovery Attacks \edit{against Ideal Ciphers}}

\author{\IEEEauthorblockN{Benjamin Harsha, Jeremiah Blocki}
\IEEEauthorblockA{\textit{Purdue University}\\
\IEEEauthorblockA{\textit{\{bharsha,jblocki\}@purdue.edu}}\\}
}

\maketitle

\begin{abstract}
It has been established that quantum algorithms can solve several key cryptographic problems more efficiently than classical computers. As progress continues in the field of quantum computing it is important to understand the risks they pose to deployed cryptographic systems. Here we focus on one of these risks - quantum key-recovery attacks against ideal ciphers. Specifically, we seek to model the risk posed  by an economically motivated quantum attacker who will choose to run a quantum key-recovery attack against an ideal cipher if the cost to recover the secret key is less than the value of the information at the time when the key-recovery attack is complete. In our analysis we introduce the concept of a quantum cipher circuit year to measure the cost of a quantum attack. This concept can be used to model the inherent tradeoff between the total time to run a quantum key recovery attack and the total work required to run said attack. Our model incorporates the time value of the encrypted information to predict whether any time/work tradeoff results in a key-recovery attack with positive utility for the attacker. We make these predictions under various projections of advances in quantum computing. We use these predictions to make recommendations for the future use and deployment of symmetric key ciphers to secure information against these quantum key-recovery attacks. We argue that, even with optimistic predictions for advances in quantum computing, 128 bit keys (as used in common cipher implementations like AES-128) provide adequate security against quantum attacks in almost all use cases.
\end{abstract}

\section{Introduction}\label{sec:intro}

As the field of quantum computing progresses it is crucial for security practitioners to understand the potential risks posed to deployed cryptosystems. In this work, we focus on quantum key-recovery attacks for symmetric key \edit{ciphers}, e.g., \edit{the Advanced Encryption Standard} (AES) \edit{blockcipher}.  Classically a symmetric key-recovery attack requires \edit{$N \approx 2^{n}$} queries in the ideal cipher model\footnote{\edit{In the ideal cipher model the attacker is only allowed to interact with the cipher as a blackbox oracle.}} where $n$ is the size of the secret key (bits). By contrast, Grover's algorithm only requires $\approx 2^{n/2}$ queries in the (quantum) ideal cipher model. While \edit{Grover's} attack requires exponential work\footnote{By contrast, Shor's algorithm can be used to break any public key encryption scheme whose security relies on the hardness of the integer factorization or discrete logarithm problem in polynomial time. This includes most widely deployed public key encryption/signature schemes including RSA, EC-DSA, Schnorr Signatures, ECDH etc. Thus,   there is a need to migrate towards ``Post-Quantum'' schemes that resist known quantum attacks like Shor's algorithm~\cite{nist_call_for_pqc,chen2016report}.}, it constitutes a dramatic reduction compared to classical attacks. For example, it would be infeasible for a powerful nation state attacker to evaluate AES $2^{128}$ times, but $2^{64}$ evaluations might be feasible even for much less sophisticated attackers! 

Traditional wisdom says that one can ensure $n$ bits of security for an ideal cipher by simply selecting $2n$ bit keys instead of $n$. However, this conservative advice might dramatically overestimate the capability of the attacker. In particular, Grover's search requires $2^{n/2}$ {\em sequential} queries meaning that attack might not finish in our lifetime. We remark that, in the ideal cipher model, any quantum key-recovery attack making at most $O\left(2^{n/2}/\sqrt{k}\right)$ sequential queries requires at least $\Omega\left(2^{n/2}\sqrt{k}\right)$ total queries\edit{, where $k$ represents the number of parallel quantum circuits being used}. Thus, one can parallelize Grover's search to reduce the running time by a factor of $\sqrt{k}$, but this approach necessarily increases the total amount of work by a factor of $\sqrt{k}$.

In this paper we advocate for an economic approach to evaluate the security of symmetric key primitives (e.g., AES-128) in a post-quantum world instead of focusing only on the running time of the fastest attack. Wiener argued that the ``full cost'' of a cryptanalytic attack~\cite{JC:Wiener04} should account for all of the required resources e.g., the cost of the circuit running the attack amortized over the number of instances that can be solved over the lifetime of the circuit. This view has guided the design and analysis of secure memory hard functions for protecting low entropy secrets like passwords against brute-force attacks~\cite{Per09,Argon2,STOC:AlwSer15,C:AlwBlo16,SP:BloHarZho18}. Taking this view we aim to model (and lower bound) the cost of running a quantum key-recover attack. We take the view that an attacker will only run a brute-force attack if the ``full cost'' of the attack is less than the value of the information that can be decrypted at the time when the quantum brute-force attack completes i.e., information decrypted $10$ years in the future may be worth less than if the documents had been decrypted today.

\section{Contributions}\label{sec:contributions}
We introduce an economic model to analyze the vulnerability of ideal symmetric key ciphers to quantum key-recovery attacks such as Grover's algorithm. In our model a rational attacker will run a quantum  key-recovery attack only if the cost of running the attack exceeds the value of the information at the time the decryption key is recovered.  One of the challenges an attacker faces is that Grover's algorithm is inherently sequential i.e., the algorithm runs in $\Omega(2^{n/2})$ sequential steps. Thus, the value of the encrypted information may be significantly reduced by the time the symmetric key is recovered e.g., after $10$ to $100$ years. Our economic model incorporates several different models of the time-value of information, as well as the full space of time/cost trade-offs available to a quantum attacker. For example, the attacker can reduce the running time by a factor of $\sqrt{k}$ by running $k$ independent searches in parallel, but this increases the full cost of Grover's attack by a factor of $\sqrt{k}$. It has been shown that (parallel) Grover search is asymptotically optimal meaning that {\em any} quantum key-recovery attack will run up against the same fundamental trade-offs~\cite{zalka1999grover}. 

Given a concrete implementation of the cipher as a quantum circuit (depth/width), predictions about the speed and cost of future quantum computers and a model describing the time-value of the encrypted information our economic model allows us to quickly decide whether or not it is profitable for an attacker to run a key-recovery attack. Similarly, we can use our model to predict how fast/cheap a quantum computer would need to be to make a key-recovery attack profitable.  It is impossible to predict how quickly quantum technology will advance. Instead we consider a wide range of predictions about the speed, size and cost of quantum computers in the next few decades and analyze costs in each scenario. We use {\em quantum mania} to refer to the world in which all of our most optimistic predictions about advances in quantum computing are realized\footnote{The name {\em quantum mania} is intentionally meant to be reminiscent of {\em crypto mania}~\cite{cryptomania} in Impagliazzo's ``five worlds.''}. In each world we can use an implementation of a quantum cipher circuit to calculate an estimated monetary value for a Cipher Circuit Year (CCY) which intuitively represents the amortized annual cost of using a quantum circuit (e.g., for the AES cipher).

As a case study we consider an organization that is considering deploying AES-GCM with $n$ bits of security in an embedded device with different security parameters $n \in \{128,196,256\}$. On the one hand it would be possible to decrease energy consumption (and/or increase throughput) by using $n=128$ bit keys~\cite{HelionAES}. On the other hand the lifetime of an embedded device can be several decades meaning that the organization will want to ensure that the cryptosystem is sufficiently resilient against quantum computers several decades in the future. It would be wise for the organization to make a conservative decision to ensure that the cost a future quantum key-recovery attack is prohibitive even under optimistic predictions about the advancement of quantum computing. We argue that a conservative organization can safely use AES-128. In particular, even under the most optimistic predictions about the speed/cost of future quantum computers (quantum mania) the cost of cracking a 128-bit key within 100 years would be {\em at least} $\$ 9.8 \times 10^{10}$ (100 Billion USD) or at least $\$9.8 \times 10^{11}$ (1 Trillion) to crack the key within 10 years. We conclude that AES-128 should provide sufficient protection against a rational attacker in almost all application scenarios. 

\section{Background}\label{sec:background}
While a deep understanding of quantum computing is not required to understand the results in this paper we include some basics in Appendix~\ref{apdx:qintro}. This background includes standard models/notations for quantum computing and a description of how Grover's algorithm works.

\subsection{Grover's algorithm}\label{sec:groverintro}
Grover's algorithm~\cite{STOC:Grover96} is (in Grover's words) ``A fast quantum mechanical algorithm for database search". Given black box access to some function $f: X \rightarrow Y$ and some value $y \in Y$ we wish to find a value for $x^* \in X$ such that $f(x^*) = y$. Using a classical computer requires $O(N)$ time, where $N = |X|$. However, if we exploit Grover's algorithm this can be improved from $O(N)$ to $O\left(\sqrt{N}\right)$. \edit{Given a known plaintext/ciphertext pair $(m,c=Enc_K(m))$ we can define a function $f_{c,m}(K')$ which outputs $1$ if and only if $Dec_{K'}(c)=m$ i.e., if a candidate encryption key $K'$ converts the known plaintext into the known ciphertext. Thus, if $K \in \{0,1\}^n$ we can recover the decryption key $K$ in time $O\left(\sqrt{2^n}\right)$ by running Grover's search on inputs $y=1$ with the function $f_{c,m}$.} From an asymptotic standpoint this significantly decreases the amount of time it would take to brute force a key e.g. it would only take $\approx 2^{128}$ steps to find a key for AES with a 256 bit key. This is especially significant for AES when using shorter key lengths (e.g. 128 bits) where the number of steps required to find a key becomes more and more feasible. Further details on Grover's algorithm can be found in Appendix~\ref{apdx:gintro}.

Grover's algorithm can be slightly modified to run with $k$ ``buckets" in parallel. The basic idea is to partition the search space $X$ into $k$ buckets $X_1,\ldots, X_k$ of size $N/k$ and run Grover's search on each bucket in parallel. Each bucket $X_i$ has size $\frac{N}{k}$ so we make $\sqrt{N/k}$ queries in each bucket $X_i$ a total of $\sqrt{Nk}$ queries over all $k$ buckets. Thus, we obtain a speedup of $\sqrt{k}$ but our total costs also increase by  $\sqrt{k}$. Zalka~\cite{zalka1999grover} proved that this is optimal. In particular, {\em any} quantum search algorithm running in sequential time $\sqrt{N/k}$ must make {\em at least} $\Omega(\sqrt{Nk})$ queries to our function $f$. Thus, in the ideal cipher model Grover's search (and its parallel counterparts) are optimal~\cite{zalka1999grover}. We observe that if running time is not an issue then it is always more cost effective to set $k=1$.

\subsection{Current and Future Quantum Computers}
While quantum algorithms like Shor's and Grover's have been known for some time as of this writing there are not yet quantum computers capable of running them in a practical attack. In fact, there is a very significant gap between the number of qubits in current quantum computers and the number of qubits that would be required to run an effective attack. For much of the past decade systems were limited to only a handful of qubits (e.g. 2-10)~\cite{corcoles2013process,barends2014superconducting} though in the past few years the number of qubits has jumped to the order of tens of qubits per quantum computer~\cite{hsu_2018,kelley_2018}. While the Bristlecone project was recently able to achieve ``quantum supremacy'' in the sense that they performed a computation faster than any classical computer could\cite{kelley_2018}. However, we are still quite far away from having a quantum computer capable of running an attack like Grover's algorithm or Shor's algorithm. We note that while companies like D-Wave claim to have much larger quantum computers these are not universal quantum computers capable of running attacks like Shor or Grover \cite{temperton_2017}.

	Estimating what future quantum computers will look like in the near and distant future is a difficult task made harder by the low number of existing data points. We take an approach similar to Impagliazzo and examine multiple possible ``worlds" with different levels of future advancement \cite{cryptomania}. These worlds range from fairly steady improvements to incredible advances in the field. There are significant technical barriers to building full scale quantum computers e.g., decoherence~\cite{knight_2017}, temperature maintenance~\cite{dwaveDescription} and high costs. Current quantum systems require temperatures very close to absolute zero and are only stable for tens of microseconds. In our most optimistic world (quantum mania) we assume that all of these challenges are convincingly addressed making it feasible to build cheap/fast quantum computers which are stable for years.

\section{Related Work}\label{sec:related}

\subsection{Analysis of Grover's algorithm}
From its introduction in 1996~\cite{STOC:Grover96} Grover's was recognized as a clear and practical example of a quantum algorithm that can outperform classical algorithms. Bennett et. al. showed that Grover's search was asymptotically optimal and that $NP$ cannot be solved by a quantum Turning machine (QTM) relative to a uniformly random oracle~\cite{bennett1997strengths}. In 1999 Zalka conducted a more fine grained analysis showing that for any number of oracle lookups up to $\frac{\pi}{4}\sqrt{N}$ Grover's algorithm is optimal~\cite{zalka1999grover}. Zalka also showed that the parallel version of Grover's algorithm is optimal.

Grassl et al~\cite{PQCRYPTO:GLRS16} looked into the problem of constructing a concrete quantum AES circuit motivated by quantum key-recovery attacks. In their analysis they provide concrete numbers (width/depth) for various AES implementations. Scott Fluhrer considered the problem of running a key-recovery attack for AES with a fixed computation budget (\# AES queries) or with a fixed time bound. For example, he estimated that because Grover's algorithm is inherently sequential it would take at least $2^{125}$ entangled queries to our AES-192 to recover the $192$-bit key within $200$ years. By contrast, we focus on developing an economic model to analyze the cost/benefit of a rational attacker and demonstrate that AES-128 should be safe against rational attackers.

\subsection{Post-quantum cryptography}

Shor's algorithm can be used to break any public key encryption scheme whose security relies on the hardness of the integer factorization or discrete logarithm problem in polynomial time. This includes most widely deployed public key encryption/signature schemes including RSA, EC-DSA, Schnorr Signatures, ECDH etc. Thus,  there is a need to migrate towards ``Post-Quantum'' schemes that resist known quantum attacks like Shor's algorithm~\cite{nist_call_for_pqc,chen2016report}.  The U.S. National Institute of Standards and Technology (NIST) has been working on developing a set of standards for post-quantum cryptography. NIST first published a report on quantum cryptography in 2016~\cite{chen2016report} which outlined their understanding and future plans, and released a call for proposals in 2017~\cite{nist_call_for_pqc}. In this document NIST proposes that an attacker running an attack over one or ten years be limited to a quantum circuit of $2^{40}$ or $2^{64}$ layers, respectively. This allows us to implicitly derive speeds for quantum computers running these attacks.

\subsection{Modeling Economic Attackers and their Costs}
Multiple authors have argued that space-time (or area-time) product is a more appropriate measure of costs than time alone, including Wiener's notion of ``full cost"~\cite{JC:Wiener04}, Alwen and Serbinenko's amortized area-time cost~\cite{STOC:AlwSer15}. In the password hashing competition~\cite{PHC} all but one finalist claimed some form of memory hardness i.e., high area-time cost. Our economic model is inspired by a model Blocki et al.~\cite{blocki2016cash,SP:BloHarZho18} developed to analyze the behavior of a (profit motivated) password cracking attacker. Space-time costs have been applied in quantum computing to show that quantum hash collision methods are not as cost effective as classical circuits \cite{bernstein2009cost}, analyze the cost of RSA factorization~\cite{gidney2019factor}, AES implementations~\cite{PQCRYPTO:GLRS16}, and elliptic curve cryptography~\cite{AC:RNSL17}. Especially relevant to this work are estimates for the width and depth of AES (\cite{PQCRYPTO:GLRS16,almazrooie2018quantum,langenberg2020reducing}).

\section{An Economic Model for Quantum Key-Recover Attacks}\label{sec:econ}
We now introduce an economic model that estimates the gain (or loss) of a quantum key-recovery attack. Our model includes the following components: (1) The initial value $v_0$ of the encrypted information and a function $R(T,v_0)$ which describes how this value decays over time $T$. (2) A time limit $T_y$ (years) for the attack e.g., $1$--$100$ years. (3) The width and depth of a quantum circuit implementing the cipher we are analyzing e.g., see~\cite{PQCRYPTO:GLRS16,cryptoeprint:2019:854} for estimates of AES. (4) The (predicted) speed of a universal quantum computer (gates/sec), (5) The (predicted) cost of renting a single quantum circuit capable of evaluating this cipher (dollars/year). Given these parameters our model allows us to determine whether or not a profitable attack exists. Fixing all of the parameters except for the initial value of the encrypted information we can determine how valuable the information would need to be for a quantum key-recovery attack to be profitable. Alternatively, fixing the initial value of the information (and a decay function) we can ask how fast/cheap a quantum computer must be to make a quantum key-recovery attack profitable.

We remark that components three and four of our model  (speed/cost of future quantum computers)  are arguably the most difficult to predict. We advocate for a conservative approach where we attempt to upper bound (resp. lower bound) the speed (resp. cost) of a future quantum computer. We remark that NIST considers $2^{64}$ to be a safe upper bound on the depth of any quantum circuit which can be evaluated in $10$ years which would correspond to a speed of $5.8\times 10^{10}$ gates per second. Thus, we might take $60$ GHz as our conservative upper bound on the speed of a quantum computer. 

The attacker will select a desired time $T$ (years) for the key-recovery attack to complete. We can infer the level of parallelism necessary to complete the attack in time $T$ given additional information about the depth of our quantum circuit implementing our cipher (e.g., AES) as well as the gate propagation speed of our quantum computer. We use $C(T)$ to denote the minimum possible cost of a quantum key-recovery attack with a time bound $T$. Intuitively, as $T$ decreases the level of parallelism increases as well as the cost $C(T)$. We use the reward function $R(T,v_0)$ to describe the attacker's benefit when the encrypted information is recovered at time $T$. Here, $v_0 = R(0,v_0)$ denotes the initial value of the encrypted information which may decrease over time. Thus, the profit of the attacker is $P(T,v_0)= R(T,v_0)-C(T)$ and the attacker will select the time parameter $T^* = \arg\max_T P(T,v_0)$ to optimize profit. If $P(T^*,v_0)<0$ a rational attacker will choose not to attack. For notational convenience we use  $P(0,v_0)=0$ to denote the profit of an adversary who does not run the attack i.e., $T=0$.

\subsection{Cipher Circuit Year}
To estimate the costs of running an attack we first define the concept of a Cipher Circuit Year (CCY). Intuitively, a CCY represents the annual rental cost (which factors in equipment, labor, electricity, and any other expenses) of a quantum computer capable of evaluating our cipher (e.g., AES)\footnote{Alternatively, we could think of CCY as representing the opportunity cost when this quantum computer used to running our key-recovery attack instead of performing other computation. Finally, we could think of CCY as representing this cost of building the quantum computer divided by the (expected) number of years before the quantum computer breaks.}. We can use CCY as a way to examine the monetary cost of a key recovery attack. For example, if we are able to complete a key recovery attack (e.g using Grover's algorithm) with no parallelism (i.e. using only one circuit) in 10 years then this attack would cost 10 CCY. However, if the same attack was completed in 1 year (which will require the use of 100 circuits running Grover's algorithm in parallel) we would have a cost of 100 CCY. Similar notions such as full cost~\cite{JC:Wiener04} or aAT complexity~\cite{STOC:AlwSer15} have been very fruitfully applied in the area of password hashing as a method of estimating costs of computation. Throughout this work we are considering attacks in the (quantum) ideal cipher model i.e. we do not concern ourselves with (quantum) structural attacks against a cipher like AES e.g. \cite{bonnetain2019quantum}.

\subsection{Required Level of Parallelism and Attack Costs} Suppose that we have a time bound $T_y$ (unit: years) for our key-recovery attack. Given the gate propagation speed $s$ (Hz) of a quantum circuit we can use $T_y$ to upper bound the total depth $t = T_y \times s$ of our computation (quantum gates) e.g., if $s=1GHz$ and $T_y=1$ year then $t=3.15 \times 10^{16}$ seconds. Supposing that our cipher can be implemented as a depth $d$ circuit our key-recovery attack can make at most $t/d$ sequential oracle queries to the cipher. If we partition our search space $\{0,1\}^n$ into $k$ buckets of size $N/k$ and run Grover's attack on each bucket in parallel then we require at least $\frac{\pi}{4}\sqrt{\frac{N}{k}}$ sequential oracle calls in each bucket ($\frac{\pi}{4}\sqrt{Nk}$ total oracle calls). Thus, $t/d \geq \frac{\pi}{4}\sqrt{\frac{N}{k}}$ which means that we require parallelism $k\geq\frac{\pi^2 N}{16\left(\frac{t}{d}\right)^2}$. The total cost will be minimized when equality holds. The total cost will be $C(T_y) = T_y \times k \times C_{CCY}$, where $C_{CCY}$ is the cost of a CCY e.g., in USD. We remark that the value of $k$ will depend on the time bound $T_y$, the depth $d$ of our cipher and the speed $s$ (Hz) of our quantum computer. Substituting into the above formula we get
$$
C(T_y) = \frac{C_{CCY} \pi^2 N d^2}{16 T_y s^2}
$$

Intuitively, the cost decreases as we relax the time bound $T_y$. If $T_y \geq \frac{\pi}{4}\sqrt{N} * \frac{d}{s}$ is sufficiently large to set $k=1$ we have $C(T_y) = C\left( \frac{\pi}{4}\sqrt{N} * \frac{d}{s}\right)$.

We note that attack costs are directly linked to an attacker's strategy. If an attacker considers the value of information to be less than the cost to run the attack we say that a rational attacker will choose to not run the attack, leaving the information secure.

\subsection{Time-Value of Information and Reward Functions}\label{sec:timeval}
We first discuss several different instantiations of the reward function  $R(T,v_0)$ which defines the time-value of the encrypted information. We will always assume that the function is monotonic i.e., $R(T,v_0) \leq R(T-\epsilon, v_0)$. Intuitively, obtaining the secret information earlier (e.g., at time $T-\epsilon$) is preferable to obtaining the secret information later\footnote{If the attacker prefers to wait to time $t$ to recover the secret information he could always run the attack and then wait $\epsilon$ seconds to measure the quibits}. In our analysis we consider three types of reward functions: (1) Constant functions $R(T,v_0)=v_0$ i.e., the time-value of the information does not diminish over time. (2) Threshold Functions where the information has value $v_0$ before time $T'$ and value $0$ afterwards i.e.,  $R_{T'}(T,v_0) = v_0$ whenever $T < T'$ and  $R_{T'}(T,v_0) = 0$ for $T > T'$. (3) Delta Discounting where the time-value of the information smoothly decays with some fixed rate $0<\delta < 1 $ i.e., $R_{\delta,T'}(T,v_0) = v_0\delta^{T}$. While this is not an exhaustive list of all possible reward functions we believe our list constitutes a reasonable range of behaviors.  

We remark that a threshold function is appropriate in settings where the encrypted information will become public at some time $t'$ in the future e.g., scripts for a soon to be released blockbuster movies or plans for an upcoming military campaign. The constant reward function can be seen as a special case of delta-discounting with $\delta = 1$ and threshold $T'=\infty$. Below we analyze the attacker's optimal strategies with respect to each reward function.

\subsection{Rational Attacker Strategies} A symmetric key-recovery attacker can pick a desired parallelism parameter $k$. Larger values of $k$ reduce the running time $T$ . Thus, by picking large $k$ we can potentially earn a larger reward $R(T,v_0)$, but at the expense of total cost $C(T)$. However, as long as the total profit $P(T,v_0) = R(T,v_0)-C(T)$ increases it is in the adversaries best interest to pick a larger value of $k$.  \\

{\bf Constant valuation: }For constant reward functions profit is maximized whenever $C(T,k)$ is minimized. As our total time and work only increase with the addition of more oracles, $C$ is minimized by setting $k=1$ i.e. running a sequential attack. We argue that constant valuation is rarely an appropriate model e.g., we expect that the value of information will not be useful after $100$ years since most people who are currently alive won't be around to benefit. \\

{\bf Threshold function: } We next consider the threshold reward functions where information has value $v_0$ before time $T'$ and value $0$ afterwards e.g., plot points for an upcoming movie. 
\[ R_{T'}(T,v_0) = \begin{cases} 
      v_0 & T \leq T' \\
      0 & T > T' 
   \end{cases}
\], where $v_0$ is the value of the information if it is recovered in time. In such a case there is no need to decrypt the information after time $T'$ so the attacker effectively faces a time limit of $T'$. Since the reward is constant before time $T'$ the attacker will maximize profit by selecting the minimum possible level of parallelism necessary to finish in time exactly $T'$ i.e., $k = \frac{\pi^2 N}{16 \frac{t}{d}^2}$ where  $t = T' \cdot s$. \\

{\bf Delta discounting with Threshold} We now analyze the behavior of the attacker with smooth $\delta$-discounting reward functions i.e., $R_{\delta,T'}(T,v_0) = v_0\delta^{T}$ for $T \leq T'$ and $R_{\delta,T'}(T,v_0) = 0$ if $T \geq T'$. Here, $0 < \delta \leq 1$ is our decay parameter and $T'$ is our threshold.  The attacker wants to pick a time $T$ which maximizes profit $P(T,v_0)= R_{\delta,T'}(T,v_0) - C(T)$. We show that there are three possible ways to maximize the profit function $P(T,v_0)$. (1) If the attacker does not run the attack $T= 0$ then $P(0,v_0)=0$. (2) The attacker sets $T=\min\{T_{seq},T'\}$ where $T_{seq} = \frac{\pi d}{4s}\sqrt{N}$ is the time to run the sequential version of Grover's algorithm $(k=1)$ when the speed is $s$ and the depth of the underlying cipher circuit is $d$. (3) The attacker sets $T = T^*$ for a special value $$
T^* = \frac{2 W\left(\frac{1}{2}\sqrt{c}\log \delta\right)}{\log \delta}.
$$ Here we let $c = \frac{\Lambda}{v \ln \delta^{-1}}$ and $W(\cdot)$ denotes the analytic continuation of the product log function i.e., the Lambert W function. We note that this function can be efficiently evaluated. The value $T^*$ is derived by analyzing the derivative of $P(\cdot,v_0)$. The full details of this derivation are in Appendix~\ref{apdx:derive}.

\subsection{On the Future Cost and Speed of Quantum Computers}
Our economic model utilizes predictions of the future speed/cost of quantum computers. However, it is difficult (or impossible) to predict what these values may be. Instead we consider a range of possible future worlds: quantum mania, optimistic improvements and steady improvements. Arguably, all of these worlds represent optimistic predictions of the future power of quantum computers. We could add a fourth pessimistic world where the field of quantum computing is stuck for decades due to insurmountable technical barriers e.g., decoherence, temperature maintenance. However, in such a world it would not be interesting to analyze quantum attacks. We advocate for a conservative approach where we attempt to upper bound (resp. lower bound) the speed (resp. cost) of a future quantum computer. In particular, if an attack is not profitable in our quantum mania scenario then it is reasonable to assume that no attack will be profitable.

\begin{itemize}
\item {\bf Quantum Mania: } Here we assume that quantum computers have enjoyed incredible advances, both in gate speed, number of qubits, and cost. In particular, we assume that quantum circuits can be evaluated at a gate propagation speed of 60GHz which we derive from NIST's proposed upper bound on the maximum depth $(2^{64})$ of a quantum circuit which could be evaluated in $10$ years~\cite{nist_call_for_pqc} i.e., \edit{$60GHz \approx 2^{64}/(10 \times Y_{sec})$ where $Y_{sec} = 3.154 \times 10^7$ is the number of seconds a year.} We also assume that dramatic advancements in QC technology  e.g. temperature maintenance and construction costs making it possible to rent a quantum AES circuit for $\$50$ per year i.e., $C_{CCY}=\$50$. 

\item {\bf Optimistic improvements: } We assume a slightly slower gate propagation speed of 1GHz for quantum computers comparable to the clock speed of current desktop computers. We also assume that $C_{CCY} = 500 USD$. This price is meant to be in line with a budget desktop one can currently purchase.
\item {\bf Steady improvements: } We assume that future quantum circuits can be evaluated at a gate propagation speed of 100MHz. We set $C_{CCY} = 50000 USD$ here. In this scenario the future speed/cost of quantum computers is not comparable to current classical machines. However, this future world would still constitute a dramatic increase in QC technology. 
\end{itemize}

\section{Case study: AES128}
In this section we use our economic model to analyze the cost of breaking a $128$ bit AES key. To apply our model we first require a concrete implementation of AES-128 as a quantum circuit. Multiple groups have consider the problem of implementing AES-128 as a quantum circuit resulting in a series of increasingly efficient constructions~\cite{PQCRYPTO:GLRS16,almazrooie2018quantum,langenberg2020reducing}. Specifically, Langenberg et al.~ \cite{langenberg2020reducing} provide the implementation with the smallest depth $d \approx 5.8 \times 10^4$. This corresponds to $3.27 \times 10^{13}$ sequential AES oracle calls per year in our quantum mania scenario. \\

In our analysis we consider an attacker with a threshold reward function $R_{T'}(T,v_0)$ for thresholds $T' \in \{1,10,100\}$ years. Here we aim to determine how valuable the encrypted information $v_0$ must be for a profitable attack to exist. We repeat this analysis for each of our quantum scenarios: quantum mania, optimistic improvements and steady improvements. Similarly, we can analyze the behavior of a profit motivated attacker when faced with $\delta$-discounting rewards $R_{T',\delta}(T,v_0)$ for thresholds $T' \in \{1,10,100\}$ years. Here, we plot the minimum reward $v_0$ for a profitable attack vs. $\delta$. Intuitively, as $\delta$ increases (slower diminishing rewards) the minimum value $v_0$ will increase. Finally, if we fix $v_0$ we can ask how fast/cheap would a quantum computer need to be for a profitable attack to exist.

\subsection{Threshold Functions} 
We first begin by examining what the costs would look like if the value follows a threshold function. When considering our 100, 10, and 1-year attacks we first convert this to some value $d$, which in this case is representing the number of circuit layers we have available given the quantum power estimates and the time available. For example, in the incredible improvements scenario, we have $t = 1.892 \times 10^{20}$, which is derived from the 60GHz figure combined with the 100-year time span. This $t$, combined with the AES-128 circuit depths from ~\cite{PQCRYPTO:GLRS16}, allows us to derive the number of oracle calls that can be made in the set time. With a number of oracle calls possible in the time we derive $k$ such that the attack would finish in the allotted time. $k$ times the attack length in years gives us our CCY cost. A final substitution for the cost ratios described earlier puts a cost in USD to run each attack. These threshold results are described in Tables ~\ref{tab:100y_threshold}, ~\ref{tab:10y_threshold}, and ~\ref{tab:1y_threshold}.
\begin{itemize}
\item {\bf 100-year attack: } A 100-year attack represents the most persistent of adversaries. This is an attack that spans generations and would represent an enormous effort to recover some piece of information. In many ways, this is an impractical attack, as there are not many cases where it is worth protecting information for 100 years. Still, we include this type of attack to make a point about the costs of a quantum key-recovery attack. The estimated costs for this attack with a threshold function are in Table \ref{tab:100y_threshold}, \ref{tab:10y_threshold}, and \ref{tab:1y_threshold}.
\begin{table}[H]\caption{100 Year Attack, Threshold value function \label{tab:100y_threshold}}
\begin{center}
\begin{tabular}{|l|c|c|}
\hline
Advancement & $t$ & $k$ \\\hline
Mania & $1.892\times10^{20}$ & $1.962\times10^7$ \\\hline
Optimistic & $3.154 \times 10^{18}$ & $7.064 \times 10^{10}$ \\\hline
Steady & $3.154 \times 10^{17}$ & $7.064 \times 10^{12}$ \\\hline
& Cost(CCY) & Cost(USD) \\\hline
Mania & $1.962\times10^{9}$ & $9.810 \times 10^{10}$\\\hline
Optimistic & $7.064 \times 10^{12}$ & $3.532 \times 10^{15}$\\\hline
Steady & $7.064 \times 10^{14}$ & $3.532 \times 10^{19}$\\\hline
\end{tabular}
\end{center}
\end{table}

\begin{table}[H]\caption{10 Year Attack, Threshold value function \label{tab:10y_threshold}}
\begin{center}
\begin{tabular}{|l|c|c|}
\hline
Advancement & $t$ & $k$ \\\hline
Mania & $1.89 \times 10^{19}$ & $1.962 \times 10^{9}$ \\\hline
Optimistic & $3.154 \times 10^{17}$ &$7.064 \times 10^{12}$ \\\hline
Steady & $3.154 \times 10^{16}$ & $7.064 \times 10^{14}$ \\\hline
& Cost(CCY) & Cost(USD) \\\hline
Mania & $1.962 \times 10^{10}$ & $9.810 \times 10^{11}$\\\hline
Optimistic & $7.064 \times 10^{13}$ & $3.532 \times 10^{16}$\\\hline
Steady & $7.064 \times 10^{15}$ & $3.532 \times 10^{20}$\\\hline
\end{tabular}
\end{center}
\end{table}

\begin{table}[H]\caption{1 Year Attack, Threshold value function \label{tab:1y_threshold}}
\begin{center}
\begin{tabular}{|l|c|c|}
\hline
Advancement & $t$ & $k$ \\\hline
Mania & $1.89 \times 10^{18}$ & $1.962 \times 10^{11}$ \\\hline
Optimistic & $3.154 \times 10^{16}$ & $7.064 \times 10^{14}$ \\\hline
Steady & $3.154 \times 10^{15}$ & $7.064 \times 10^{16}$ \\\hline
& Cost(CCY) & Cost(USD) \\\hline
Mania & $1.962 \times 10^{11}$ & $9.810 \times 10^{12}$\\\hline
Optimistic & $7.064 \times 10^{14}$ & $3.532 \times 10^{17}$\\\hline
Steady & $7.064 \times 10^{16}$ & $3.532 \times 10^{21}$\\\hline
\end{tabular}
\end{center}
\end{table}
\item {\bf Ten-year attack: } A ten-year attack still represents a fairly long-term attack. Essentially, this is an attack against information that is not time-sensitive, which might include something like bank account access credentials. The estimated cost with a threshold function can be found in Table \ref{tab:10y_threshold}.

\item {\bf One year attack} Here we consider attacks that may be of interest to an adversary wanting information that is valuable in the near future. This might include things like business strategies, financial plans, etc. The estimated cost with a threshold function can be found in Table \ref{tab:1y_threshold}.

\end{itemize} 

\subsection{$\delta$ discounting method}
For the $\delta$ discounting method we present the information in a slightly different way. We look for the minimum value such that an attack would be profitable to run. As we have $P(T_Y,v_0) = R_{T',\delta}(T_y,v_0) - C(T_y) = v\delta^{T_y} - c_{CCY}\frac{\pi^2 N T_y}{16(t/d)^2}$, where $t=T_y \cdot s$. We have a viable attack if $v \geq \frac{\pi^2 N d^2}{16 T_y s^2 \delta^{T_y}}$. Here we can set $T_y$ appropriately, as we did in the threshold experiments. When this is set we have the option to let $\delta$ range from 0 to 1, or equivalently to allow $\delta^{T_y}$ to range from 0 to 1. 
For the sake of demonstration we will be examining this kind of attack in the ``incredible improvements" scenario. To find this point $v$ we first require the conversion factor between $T_y$ and $t$. We have $d = 57894$ as our circuit depth for AES~\cite{langenberg2020reducing} (taking their round-depth estimates for a full Grover's attack). Supposing that $s= 6\times10^{10} Hz$ (quantum mania) we are able to evaluate a circuit of depth $t= T_y\times 6\times10^{10}Hz$ where $T_y$ is given in years e.g., 1 year, 10 year, or 100 years (here we take 100). Substitution gives \edit{ $\frac{t}{d} = \frac{T_y \times Y_{sec} \times 6\times10^{10}Hz}{57894} = T_y \times 3.271 \times 10^{13}$}. We know that an attack is profitable only if
$$
v \geq \frac{C_{CCY} \pi^2 N}{16 (3.271 \times 10^{13})^2 T_y \delta^{T_y}}
$$

Fixing $T_y \in \{1,10,100\}$ and letting $\beta = \frac{C_{CCY}\pi^2 N}{16\left(3.271 \times 10^{13}\right)^2 T_y}$ we have $v \geq \frac{\beta}{\delta_t}$. We plot the minimum value required for this to be true for the given $T_y$, and the results for the quantum mania world are shown in Figure \ref{fig:min_value}. We note that the case $\delta^t = 1$ is identical to the case where $\delta = 1$, and that these values will match those from the threshold function. For other values in this chart, $\delta^{T_y}$ can be understood as the remaining value at the \textit{end} of the attack e.g. $\delta^{T_y}$ of 0.2 in the 100 year case means that 20\% of the value remains after 100 years, while in the 1 year case $\delta^{T_y}$ of 0.2 means only 20\% of the value remains after a single year. In the case $\delta^{T_y} = 0.2$ the specific value of $\delta = \left(0.2\right)^\frac{1}{T_y}$. Thus, our model predicts that AES128 provides sufficient protection provided that the initial value of the information is under $\approx 10^{11.5}$ USD.

\begin{figure}
\caption{\label{fig:min_value}}
\begin{tikzpicture}[scale=0.7]
\begin{semilogyaxis}[xmin = 0.01, xmax = 1, samples=100, xlabel=$\delta^T_y$, ylabel=$v(USD)$,title=Minimum value required for positive profit. \\$\delta^{T_y}$ represents the value left after $T_y$ years]
	\addplot[blue, thick,domain=0.01:1] {9.81e10*(1.0/x)};
	\addlegendentry{100 year attack}
	\addplot[red, thick,domain=0.01:1]{9.81e11*(1/x)};
	\addlegendentry{10 year attack}
	\addplot[green, thick,domain=0.01:1] {9.81e12*(1/x)};
	\addlegendentry{1 year attack}
\end{semilogyaxis}
\end{tikzpicture}
\end{figure}

\subsection{Improvements in circuit width and depth}
We now consider the following question. What happens if we are able to develop smaller quantum circuits to compute a cipher (e.g. the quantum AES circuit)? Improvements might be made either by reducing the width or the depth of the circuit e.g., by exploiting feasible time-space tradeoffs for the function. We first note that if the width is reduced by some factor $c$ then the cost of running the attack also drops by the same factor $c$. More interesting is the case where we are able to reduce the depth of our cipher circuit. As we reduce depth $d$ of the quantum circuit (holding width constant) an attacker running at the same gate speed is able to make more cipher queries in the same amount of time. Thus the attacker saves cost on multiple fronts --- the circuit itself is smaller by some factor $c$ and the attacker is able to reduce parallelism because s/he can make more sequential queries in the same amount of time. In fact, improvements in circuit depth offer quadratic improvements in attack cost, meaning that it is worth reducing the depth of a quantum circuit even if it comes at the cost of increasing the width by a sub-quadratic amount. Specifically, if we reduce $d$ by some factor $0 < \beta < 1$ we now have a more relaxed requirement for our time $t$. Previously for some set real time limit $T_y$ we had time to make $xt$ oracle calls. We now have time for $\frac{1t}{d\beta} > \frac{t}{d}$ calls. Thus our previous $k = \frac{\pi^2 N}{16 \frac{t}{d}^2}$ becomes $k' = \frac{\pi^2 N}{16 (\frac{t}{d} / \beta)^2} = \beta^2 k$. Thus a reduction by $\beta$ offers quadratic returns. Note that increases in QC speed identically affect how many calls per year we can make and offers the same tradeoff. 
Suppose that we reduce depth by a factor of 10 (as is claimed from \cite{PQCRYPTO:GLRS16} to \cite{cryptoeprint:2019:854}), holding width constant. This allows for a 100 fold reduction in the level of parallelism, and a 100 fold reduction in costs. In our 100 year quantum mania example from Table~\ref{tab:100y_threshold} this could lower attack cost to $\approx 9.8 \times 10^{8}$ USD. Far smaller than the previous estimates, but still infeasible in practice. 

\subsection{Computers capable of running profitable attacks}

Under the assumptions from our three worlds of quantum development we found that any quantum key recovery attack for a 128-bit key is not economically feasible. Here we seek to answer a related question: If we want an economically feasible attack, what kind of quantum computer would be required? We follow a similar strategy of proposing three attacks, but note that this analysis works for any relevant parameters. We begin here not by assuming any level of advancement in quantum computing but by assuming an attacked values some piece of information at a particular level. Here we select some USD amount e.g. 100,000, 1,000,000, or 10,000,000 as the value of information. We also allow the attacker to select 1, 10, and 100 year attacks in the same manner as in Section \ref{sec:econ}.
When we set the cost and time limit for these attacks we arrive not at a single quantum computer that would suffice to run the attack but a family of quantum computers with varying speeds and costs. This fact arises from the theoretical ability to bring cost per cubit down if we allow for a computer to have a higher clock speed. So, when we set the total budget and time limit for an attack we arrive at some family of computers described in terms of the cost/speed tradeoffs. These tradeoffs are subject to the quadratic increase in cost seen when increasing the level of parallelism.
We denote a family of quantum computers $\mathcal{Q}_{b,T_y,n}$ based on a budget $b$, time limit in years $T_y$, and key length $n$. This describes the set of quantum computers capable of running a quantum key-recovery attack with the relevant restrictions. The family contains all quantum computers satisfying the property: $\mathcal{Q}_{b,T_y,n} = \left\{q : C_{CCY} \leq \frac{16b\left(\frac{s}{d}\right)^2}{\pi^2 2^n T_y}\right\}$ where $d$ is the depth of the oracle circuit and $s$ is the number of circuit layers that can be processed in the given time limit. 
We can now begin to look at some families of quantum computers based on some reasonable attack budgets. Consider an attack that is of vital importance - where an attacker is willing to spend USD 100 million on an attack on AES-128, and needs it within 100 years. A quantum computer capable of doing so is in $\mathcal{Q}_{1.0 \times 10^8,100,128}$, which contains all quantum computers such that $C_{CCY} \leq \frac{1.6 \times 10^9 \left(\frac{s}{57854}\right)^2}{\pi^2 2^{128} 100} = 1.423 \times 10^{-42} s^2$. Consider a case where we would like $C_{CCY} \leq USD 1000$. This would require the computer to run at a speed of $s = 2.65 \times 10^{22}$ in 100 years, corresponding to a gate propagation speed of $8.403 \times 10^{12}Hz$, well beyond NIST's estimates of around 60GHz~\cite{nist_call_for_pqc}. Any required parameters might be inserted here to see what would be required to make the existing parameters work - a plot of the family that can solve this problem is shown in Fig \ref{fig:speedcosttradeoff}. We note that no matter which parameters you pick for this attack you end up with a computer that is either impossibly fast or impossibly cheap, meaning that no quantum computer that can run an attack with this requested budget and time limit is feasible.

\begin{figure}
\begin{tikzpicture}
\usepgfplotslibrary{fillbetween}
\begin{loglogaxis}[
	xlabel=$a (Hz)$,
	ylabel=$C_{CCY} (USD)$,	
	title=Speed vs QC CCY cost,
	enlarge x limits=false,
	axis on top,
]
\addplot[
	color = blue,
	domain=1e10:1e50,
	fill = blue!30!white
]{1.423e-42 * x^2} -- (rel axis cs:1,0) -- (rel axis cs:0,0);

\end{loglogaxis}
\end{tikzpicture}
\caption{Possible values of $C_{CCY}$ based on estimated QC speed ($b = 1.0 \times 10^8, T_y = 100, n=128$\label{fig:speedcosttradeoff}}
\end{figure}
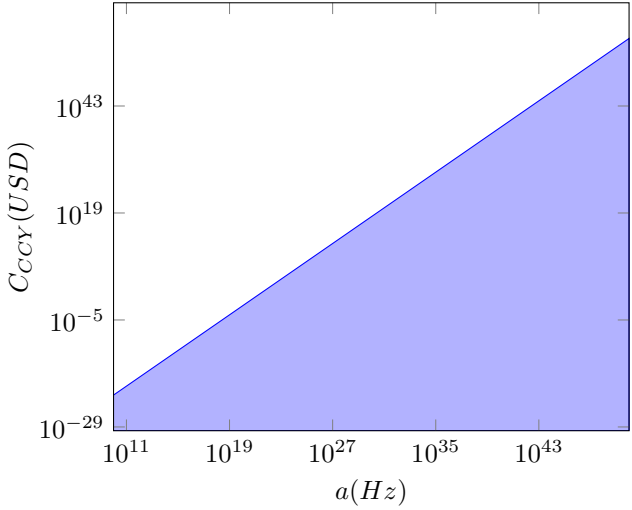


\section{m to 1 Key Recovery Attacks}
With a chosen plaintext attack where multiple keys have been used and a single chosen plaintext can be used with multiple keys it is possible to ``batch" key-recovery attacks for a more effective attack. This might be considered where any one of $m$ keys would be sufficient for an adversary to accomplish their goals e.g. to access some specific set of data that had been sent to multiple people using multiple different keys with the same nonce. We consider a chosen-plaintext attack where an attacker has manged to obtain $M$ ciphertexts $c_1,\ldots c_M$ all encrypting the same known plaintext $m$ i.e., $c_i = Enc_{k_i}(m; r)$ where $m=m_1 || m_2$  consists of two blocks and the randomness $r$ (i.e., nonce) is the same. Such a scenario might arise if we have multiple embedded devices using a stateful mode of operation like AES-CTR with a fixed initialization vector. Modes like AES-GCM would not be susceptible to this attack as long as the nonces are selected appropriately i.e., with strong PRGs. 

Our attacker will be content to crack {\em any} of the keys $k_1,\ldots,k_m$. To run Grover's algorithm the attacker would need to implement the function 
\[
f_{k_1,\ldots,k_M}(x) = 
\begin{cases}
1 & x \in \{k_1,\ldots,k_M\} \\
0 	&	\mbox{otherwise}
\end{cases} \ .
\]

This function could be implemented as follows \[
F_{c_1,\ldots,c_M}(x) = 
\begin{cases}
1 & Enc_x(m) \in \{c_1,\ldots,c_M\} \\
0 	&	\mbox{otherwise}
\end{cases} \ .
\]
We first note that (except with negligible probability) we will have $F_{c_1,\ldots,c_M}(x) = f_{k_1,\ldots,k_M}(x)$ for all inputs $x$ i.e., because $m$ is two blocks long we expect that for each $c_i$ there is only one key $k$ (namely $k=k_i$) s.t. $Enc_k(m) = c_i$. Note that each call to $F_{c_1,\ldots,c_M}$ generates just $2$ calls to the underlying cipher-circuit to obtain $c=Enc_x(m)$ --- both of these calls can be evaluated in parallel. We then need to check whether or not $c \in \{c_1,\ldots, c_m\}$. Since we want to compute $F_{c_1,\ldots,c_M}$ on the same quantum hardware used to evaluate the cipher we require that the width of our circuit is not larger that the width of our AES circuit. When we add this restriction $F_{c_1,\ldots,c_M}$ can be evaluated on a Quantum Circuit of depth $\mathcal{O}\left(d_{AES} + \frac{Mn}{w_{AES}} \right)$ where $d_{AES} \approx 1.5 \times 10^4$ and $w_{AES} \approx 10^3$ are the depth and width of the quantum AES circuit. Since $n=128$ in our analysis, whenever $M < 10^5$ the depth of the circuit is dominated by the depth of the AES circuit. \\

Theorems \ref{thm:upper} and \ref{thm:lower} below upper and lower bound the total number of ideal cipher queries necessary to recover one out of $M$ keys.

\begin{theorem} \label{thm:upper}
There exists a $k$-parallel quantum algorithm $A^{F_{c_1,\ldots,c_M}(\cdot)}$ such that $\Pr\left[A^{F_{c_1,\ldots,c_M}(\cdot)}(x)(1^n) \in S\right] > \frac{1}{2}$ in sequential time $O\left(\sqrt{\frac{N}{kM}}\right)$ and makes $O\left(\sqrt{\frac{kN}{M}}\right)$ oracle queries, where probability is taken over selection of a random subset $S \subseteq \{0,1\}^n$ of size $M$ as well as the randomness of $A^{f_S(\cdot)}$.
\end{theorem}

\begin{theorem} \label{thm:lower}
Given any constant $c \in (0,1]$ there is no $k$-parallel quantum algorithm $A^{Enc}$ running in sequential time $o\left(\sqrt{\frac{N}{Mk}}\right)$ and making at most $o\left(\sqrt{\frac{Nk}{M}}\right)$ queries to the ideal cipher that can find an element $x \in S$ with probability $\Pr\left[A^{f_S}(1^n) \in S\right] > c$, where probability is taken over selection of a random subset $S \subseteq \{0,1\}^n$ of size $m$ as well as the randomness of $A^{f_S(\cdot)}$.
\end{theorem}

The proofs for these theorems can be found in the Appendix, and generally follow commonly used methods for other Grover's algorithm optimality results. These theorems show that when considering an $M$ key batch attack running on multiple quantum computers in parallel Grover's algorithm is an asymptotically optimal solution. This also shows that as you obtain $M$ keys to batch together you can speed up attacks by a factor of $\sqrt{M}$. This can cause some significant reductions in attack costs, bringing some attacks closer to economically feasible ranges. For example, consider a setting where the attacker has access to $M=10^6$ encryptions of the same message under different AES keys. In this case the cost of cracking one of these keys within $100$ years would be around $100$ million (USD) under our quantum mania assumption. This is still quite expensive, but significantly cheaper than the $100$ billion (USD) it would take to crack each key individually.

\section{Discussion}\label{sec:discussion}

We introduced an economic model to analyze the efficacy of quantum key-recover attacks. Our results (for threshold scenarios) are summarized in table \ref{tab:100y_threshold}, \ref{tab:10y_threshold} and \ref{tab:1y_threshold}. Within these tables consider the attacker's most optimistic scenario. Suppose that we are in the ``quantum mania'' world in which the cost/speed of quantum computers improves at a rapid pace. Further suppose that the only time restriction that the attacker faces is that the key-recovery attack must be completed within $100$ years. Even in the attacker's best case scenario the cost of a key-recovery attack is estimated at $9.81 \times 10^{10}$ USD, a very significant amount.  While this is certainly less than the expected classical cost of $9.24\times 10^{29}$ USD we still see a significant financial barrier to these attacks. Under less optimistic scenarios the attacker's costs only increase e.g., if the attacker needs to recover the key in $10$ years under our ``optimistic'' assumption on advances in quantum computing the attacker's costs will be at least $3.352 \times 10^{16}$, well beyond the capabilities of any adversary. Given that the cost of a quantum key-recovery attack is so high we argue that for almost all use cases AES-128 should remain safe in a post-quantum world. We additionally stress that the values we provide should be considered lower bounds. We have ignored many significant issues that arise for quantum computers like error correction, decoherence, and electricity costs.

We advocate for rethinking the common strategy of defending against a quantum key-recovery attack by doubling the key length. In fact, not only do we find that doubling key length is usually unnecessary, we also find that adding a constant number of bits to the key is not needed as suggested in~\cite{Fluhrer17}. In settings where computational overhead is paramount (e.g., embedded devices) and the secret is under our lowest attack cost estimate ($6.63 \times 10^9$ USD) it may be better to opt for smaller key lengths.

\paragraph{Economic Analysis of Second-Preimage Attacks} In this work we focused on quantum key-recovery attacks, such as those that might be run with Grover's algorithm against AES. However there are other similar attacks that are possible using Grover's algorithm e.g., a second-preimage attack against a hash function. In this case we run Grover's algorithm with a fixed hash output and query for elements of the domain that produce the desired output. These attacks would allow for the forging of digital signatures. We note that our models may be extended to represent these attacks as well. To run the attack some number of oracle circuit years would be required, which is a concept very close to the cipher circuit years we used here. While \cite{brassard1997quantum} shows that Grover's algorithm can be used to reduce the cost of $b$ bit hash collisions to $O\left(2^{b/3}\right)$, \cite{bernstein2009cost} claims that these will remain cost ineffective, and that classical computers outperform these methods even under optimistic assumptions for quantum computer speed.

\subsection{Future work}
We have introduced a general method of modeling quantum key-recovery attacks, focusing our attention on Grover's algorithm being used to attack AES-128. The main reasons we focused on this specific instance is because we believe it focuses on a common cipher (AES) and uses the most well-known key-recovery attack (Grover's). There are also published estimates for the circuit width and depth of AES~\cite{PQCRYPTO:GLRS16} that allow us to cleanly estimate the circuit width and depth of this attack. It would be interesting to see how the model would apply in other situations, but this would require additional work to be completed in other areas. If, for example, we wanted to run the same analysis on Triple-DES, which is still a NIST approved cipher~\cite{barker2016guideline}, we would require an analysis similar to the work done by Grassl et. al. Without this the analysis cannot progress much farther beyond a price in CCY - as estimates are not available for the conversion between CCY and actual cost. 
Additional analysis for ciphers like Triple-DES may also be of interest. Whereas AES can be modeled as an ideal cipher with key size $k$ it may not be the case that a ``Triple-AES" would be an ideal cipher with key length $3k$. What would the specifics of a quantum attack against Triple-DES or Triple-AES look like, and is there a more effective way to run the attack than a Grover's attack on a key of length $3k$? 

\edit{Finally, in our analysis we assume that the gate propogation speed of a quantum circuit is upper bounded by 60GHz~\cite{nist_call_for_pqc}. Future work might explore the potential profitability of quantum key recovery attacks under the assumption that quantum computing speeds advance similar to Moore's Law and double every few years. In classical computing Moore's Law appears to be reaching its limit~\cite{theis2017end,courtland_2015}. If quantum computing does not encounter similar barriers, then quantum key recovery attacks might eventually become profitable as long as the $\delta$-discounting parameter does not dominate the rate of progress. For example, we might consider a 100 year attack scenario with some budget and reward amount and ask at what rate quantum computer speeds need to double for a viable attack to exist. }


\bibliographystyle{acm}
\bibliography{extra,../cryptobib/abbrev3,../cryptobib/crypto}

\begin{thebibliography}{10}

\bibitem{almazrooie2018quantum}
{\sc Almazrooie, M., Samsudin, A., Abdullah, R., and Mutter, K.~N.}
\newblock Quantum reversible circuit of aes-128.
\newblock {\em Quantum Information Processing 17}, 5 (2018), 112.

\bibitem{C:AlwBlo16}
{\sc Alwen, J., and Blocki, J.}
\newblock Efficiently computing data-independent memory-hard functions.
\newblock In {\em CRYPTO~2016, Part~II\/} (Aug. 2016), M.~Robshaw and J.~Katz,
  Eds., vol.~9815 of {\em {LNCS}}, Springer, Heidelberg, pp.~241--271.

\bibitem{STOC:AlwSer15}
{\sc Alwen, J., and Serbinenko, V.}
\newblock High parallel complexity graphs and memory-hard functions.
\newblock In {\em 47th ACM STOC\/} (June 2015), R.~A. Servedio and
  R.~Rubinfeld, Eds., {ACM} Press, pp.~595--603.

\bibitem{barends2014superconducting}
{\sc Barends, R., Kelly, J., Megrant, A., Veitia, A., Sank, D., Jeffrey, E.,
  White, T.~C., Mutus, J., Fowler, A.~G., Campbell, B., et~al.}
\newblock Superconducting quantum circuits at the surface code threshold for
  fault tolerance.
\newblock {\em Nature 508}, 7497 (2014), 500.

\bibitem{barker2016guideline}
{\sc Barker, E.}
\newblock Guideline for using cryptographic standards in the federal
  government: cryptographic mechanisms.
\newblock Tech. rep., National Institute of Standards and Technology, 2016.

\bibitem{bennett1997strengths}
{\sc Bennett, C.~H., Bernstein, E., Brassard, G., and Vazirani, U.}
\newblock Strengths and weaknesses of quantum computing.
\newblock {\em SIAM journal on Computing 26}, 5 (1997), 1510--1523.

\bibitem{bernstein2009cost}
{\sc Bernstein, D.~J.}
\newblock Cost analysis of hash collisions: Will quantum computers make sharcs
  obsolete.
\newblock {\em SHARCS 9\/} (2009), 105.

\bibitem{Argon2}
{\sc Biryukov, A., Dinu, D., and Khovratovich, D.}
\newblock Argon2: new generation of memory-hard functions for password hashing
  and other applications.
\newblock In {\em Security and Privacy (EuroS\&P), 2016 IEEE European Symposium
  on\/} (2016), IEEE, pp.~292--302.

\bibitem{blocki2016cash}
{\sc Blocki, J., and Datta, A.}
\newblock Cash: A cost asymmetric secure hash algorithm for optimal password
  protection.
\newblock In {\em 2016 IEEE 29th Computer Security Foundations Symposium
  (CSF)\/} (2016), IEEE, pp.~371--386.

\bibitem{SP:BloHarZho18}
{\sc Blocki, J., Harsha, B., and Zhou, S.}
\newblock On the economics of offline password cracking.
\newblock In {\em 2018 {IEEE} Symposium on Security and Privacy\/} (May 2018),
  {IEEE} Computer Society Press, pp.~853--871.

\bibitem{bonnetain2019quantum}
{\sc Bonnetain, X., Naya-Plasencia, M., and Schrottenloher, A.}
\newblock Quantum security analysis of aes.
\newblock {\em IACR Transactions on Symmetric Cryptology\/} (2019), 55--93.

\bibitem{boyer1998tight}
{\sc Boyer, M., Brassard, G., H{\o}yer, P., and Tapp, A.}
\newblock Tight bounds on quantum searching.
\newblock {\em Fortschritte der Physik: Progress of Physics 46}, 4-5 (1998),
  493--505.

\bibitem{brassard1997quantum}
{\sc Brassard, G., Hoyer, P., and Tapp, A.}
\newblock Quantum algorithm for the collision problem.
\newblock {\em arXiv preprint quant-ph/9705002\/} (1997).

\bibitem{HelionAES}
{\sc Brief, H. T.~P.}
\newblock Aes ip cores for asic, 2014.
\newblock (Retrieved February 26, 2019).

\bibitem{chen2016report}
{\sc Chen, L., Jordan, S., Liu, Y.-K., Moody, D., Peralta, R., Perlner, R., and
  Smith-Tone, D.}
\newblock {\em Report on post-quantum cryptography}.
\newblock US Department of Commerce, National Institute of Standards and
  Technology, 2016.

\bibitem{nist_call_for_pqc}
{\sc Chen, L., Moody, D., and Liu, Y.-k.}
\newblock Submission requirements and evaluation criteria for the post-quantum
  cryptography standardization process, 2017.

\bibitem{corcoles2013process}
{\sc C{\'o}rcoles, A.~D., Gambetta, J.~M., Chow, J.~M., Smolin, J.~A., Ware,
  M., Strand, J., Plourde, B.~L., and Steffen, M.}
\newblock Process verification of two-qubit quantum gates by randomized
  benchmarking.
\newblock {\em Physical Review A 87}, 3 (2013), 030301.

\bibitem{courtland_2015}
{\sc Courtland, R.}
\newblock Gordon moore: The man whose name means progress, Mar 2015.

\bibitem{qcnotes}
{\sc de~Wolf, R.}
\newblock Quantum computing: Lecture notes, Jan 2018.

\bibitem{PHC}
{\sc et~al., J.-P.~A.}
\newblock Password hashing competition, 2015.

\bibitem{Fluhrer17}
{\sc Fluhrer, S.~R.}
\newblock Reassessing grover's algorithm.
\newblock {\em IACR Cryptology ePrint Archive 2017\/} (2017), 811.

\bibitem{gidney2019factor}
{\sc Gidney, C., and Eker{\aa}, M.}
\newblock How to factor 2048 bit rsa integers in 8 hours using 20 million noisy
  qubits.
\newblock {\em arXiv preprint arXiv:1905.09749\/} (2019).

\bibitem{PQCRYPTO:GLRS16}
{\sc Grassl, M., Langenberg, B., Roetteler, M., and Steinwandt, R.}
\newblock Applying grover's algorithm to {AES}: Quantum resource estimates.
\newblock In {\em Post-Quantum Cryptography - 7th International Workshop,
  PQCrypto 2016\/} (2016), T.~Takagi, Ed., Springer, Heidelberg, pp.~29--43.

\bibitem{STOC:Grover96}
{\sc Grover, L.~K.}
\newblock A fast quantum mechanical algorithm for database search.
\newblock In {\em 28th ACM STOC\/} (May 1996), {ACM} Press, pp.~212--219.

\bibitem{hsu_2018}
{\sc Hsu, J.}
\newblock
  https://spectrum.ieee.org/tech-talk/computing/hardware/intels-49qubit-chip-aims-for-quantum-supremacy.
\newblock {\em IEEE Spectrum\/} (Jan 2018).

\bibitem{cryptomania}
{\sc {Impagliazzo}, R.}
\newblock A personal view of average-case complexity.
\newblock In {\em Proceedings of Structure in Complexity Theory. Tenth Annual
  IEEE Conference\/} (June 1995), pp.~134--147.

\bibitem{kelley_2018}
{\sc Kelley, J.}
\newblock {\em A Preview of Bristlecone, Google’s New Quantum Processor\/}
  (Mar 2018).

\bibitem{knight_2017}
{\sc Knight, W.}
\newblock Ibm announces a trailblazing quantum machine, Nov 2017.

\bibitem{cryptoeprint:2019:854}
{\sc Langenberg, B., Pham, H., and Steinwandt, R.}
\newblock Reducing the cost of implementing aes as a quantum circuit.
\newblock Cryptology ePrint Archive, Report 2019/854, 2019.
\newblock https://eprint.iacr.org/2019/854.

\bibitem{langenberg2020reducing}
{\sc Langenberg, B., Pham, H., and Steinwandt, R.}
\newblock Reducing the cost of implementing aes as a quantum circuit.
\newblock {\em IEEE Transactions on Quantum Engineering\/} (2020).

\bibitem{Per09}
{\sc Percival, C.}
\newblock Stronger key derivation via sequential memory-hard functions.
\newblock In {\em BSDCan 2009\/} (2009).

\bibitem{scion_aes}
{\sc Perrig, A., Szalachowski, P., Reischuk, R.~M., and Chuat, L.}
\newblock {\em SCION: a secure Internet architecture}.
\newblock Springer, 2017.

\bibitem{AC:RNSL17}
{\sc Roetteler, M., Naehrig, M., Svore, K.~M., and Lauter, K.~E.}
\newblock Quantum resource estimates for computing elliptic curve discrete
  logarithms.
\newblock In {\em ASIACRYPT~2017, Part~II\/} (Dec. 2017), T.~Takagi and
  T.~Peyrin, Eds., vol.~10625 of {\em {LNCS}}, Springer, Heidelberg,
  pp.~241--270.

\bibitem{dwaveDescription}
{\sc Systems, D.~W.}
\newblock A brief introduction to d-wave and quantum computing, 2017.
\newblock
  \url{https://www.dwavesys.com/sites/default/files/D-Wave-Overview-Jan2017F3.pdf}.

\bibitem{temperton_2017}
{\sc Temperton, J.}
\newblock Got a spare 15 million? why not buy your very own d-wave quantum
  computer, Jan 2017.

\bibitem{theis2017end}
{\sc Theis, T.~N., and Wong, H.-S.~P.}
\newblock The end of moore's law: A new beginning for information technology.
\newblock {\em Computing in Science \& Engineering 19}, 2 (2017), 41--50.

\bibitem{JC:Wiener04}
{\sc Wiener, M.~J.}
\newblock The full cost of cryptanalytic attacks.
\newblock {\em Journal of Cryptology 17}, 2 (Mar. 2004), 105--124.

\bibitem{zalka1999grover}
{\sc Zalka, C.}
\newblock Grover’s quantum searching algorithm is optimal.
\newblock {\em Physical Review A 60}, 4 (1999), 2746.

\end{thebibliography}

\appendix

\section*{Quantum Computing}\label{apdx:qintro}
Quantum computing allows for computation over qubits, a quantum analog of classical bits. Each qubit has two state analogous to the classical 0 and 1 states denoted as the $\ket{0} = \left( 1 \atop 0 \right)$ and $\ket{1} = \left( 0 \atop 1 \right)$ states respectively. Rather than each qubit remaining in one classical state they are able to exist as a superposition of states $\ket{\psi} = \alpha_0 \ket{0} + \alpha_1 \ket{1}, \alpha_i \in \mathbb{C}$. Multiple qubits are combined via tensor products and are denoted in this paper as (e.g.) $\ket{01} = \ket{0} \otimes \ket{1}$.
When a quantum computer is in a state $\ket{\psi}$ it can be advanced to a new state $\ket{\psi'}$ by multiplication with a unitary matrix $U$ (a matrix whose inverse is its conjugate transpose) i.e. $\ket{\psi'} = U\ket{\psi}$. Quantum algorithms are described as sequences of unitary transformations on these qubits. For a far more detailed description of the basics of quantum computing we direct the reader to~\cite{qcnotes} or another online resource of their choosing.

The ability to operate while in a superposition gives quantum computers an advantage over classical computers - with quantum computers having a significant advantage when solving certain types of problems. In several cases this improvement is very significant e.g. the ability to solve the integer factorization and discrete log problems in polynomial time - a feat that has not been accomplished with classical computers. The ability to solve these problems efficiently has clear implications for security, allowing for quick attacks on several asymmetric or public key cryptosystems. In other cases quantum algorithms may provide a significant advantage over classical attacks, but not so significant as a polynomial time attack. Grover's algorithm, which we will examine most closely, is an example of this type of attack. Here a quantum computer is capable of providing quadratic speed up which, while not as strong as a polynomial time attack, is still of interest.

\subsection{Some relevant quantum definitions}
Throughout this paper we will be working with several specific quantum algorithms to accomplish our goals. These algorithms will make use of the following quantum gates:
\begin{itemize}
\item {\bf Hadamard gate:} $H  = \frac{1}{\sqrt{N}}\begin{bmatrix}
1 & 1 \\
1 & -1
\end{bmatrix}
$ Maps the $\ket{0}$ and $\ket{1}$ states to $\frac{1}{\sqrt{2}}\ket{0} + \frac{1}{\sqrt{2}}\ket{1}$ and $\frac{1}{\sqrt{2}}\ket{0} - \frac{1}{\sqrt{2}}\ket{1}$ respectively. Importantly - the Hadamard gate is its own inverse. 
\item {\bf Phase query:} $O_{i,\pm}: \ket{i} \rightarrow \left(-1\right)^{f(i)}\ket{i}$. Negates the amplitude if the value of $f$ at $i$ is 1.

\item {\bf Inversion around mean amplitude $U_s$:} The unitary transform $U_s = H^{\otimes n}\left(2\ket{0^n}\bra{0^n} - I_n\right)H^{\otimes n}$ inverts the superposition $\ket{\psi}$ around the mean amplitude i.e. $\alpha_{i+1} = \left(\frac{2}{2^n}\sum\limits_{j}\alpha_j \right) - \alpha_i$

\end{itemize}

\section*{Grover's Algorithm definition}\label{apdx:gintro}
\begin{enumerate}
\item Begin in the $\ket{\psi} = \ket{0}^{\otimes N}$ state
\item Apply $H^{\otimes N}$ to get $\ket{\psi} = \frac{1}{\sqrt{N}}\sum\limits_{i=0}^{N-1}\ket{i}$
\item Repeat $O\left(\sqrt{N}\right)$ times:\\
\begin{enumerate}
	\item Apply the phase query $\ket{\psi} \gets O_{\pm}\ket{\psi}$
	\item Apply $U_s$ to invert about the mean amplitude 
\end{enumerate}
\item Observe the result
\end{enumerate}

Note that because the phase query $O_{\pm}$ negates the amplitude of $\ket{x^*}$ while the mean amplitude stays positive we will end up increasing the the amplitude of $\ket{x^*}$ while decreasing the amplitude of all other states. After repeating $O(\sqrt{N})$ times the amplitude $\alpha_{x^*}$ will be likely to be observed. Here the constants are important, as after a time the amplitude $\alpha_{x^*}$ will start decreasing if the inner loop is iterated too much. When N is large (which it will be, for our purposes) the loop should be run about $\frac{\pi}{4}\sqrt{N}$ times for a high probability $\left(\approx 1 - \frac{1}{N}\right)$ of success~\cite{boyer1998tight}.

\section*{Comparison with Classical Attacks}\label{apdx:classicalcomp}
While we are primarily concerned with the economics of quantum attacks it is worth taking a moment to establish a comparison point with classical attacks. Here we once again note a few things - first is that classical attacks parallelize perfectly i.e. there is no penalty to dividing the search space and running in parallel. Because of this it is much easier to establish the costs of an attack. We also note that time limits are less relevant to classical attacks in many situations. So long as the time limit to run a key-recovery attack is longer than the expected lifespan of the equipment to run it we can generally expect the costs of a parallel attack to match the costs of a sequential attack (within some reasonable constant factors). 
For a baseline, we consider an attack using an FPGA setup for AES128 capable of making guesses at 350 million guesses per second while using 6.6W~\cite{scion_aes}. If we take a value of USD 0.08 per kWh for electricity costs we have an expected attack cost using the brute force algorithm $\mathcal{A}$ of $\mathbb{E}(cost(\mathcal{A})) = 2^{127}*cost(\mathcal{A})$. $cost(\mathcal{A}) = \frac{0.0066 * 0.08 * 60^2}{3.5*10^{8}} \approx 5.43*10^{-9}$. A final estimated cost would then be $2^{127}*5.3*10^{-9} \approx 9.24*10^{29}$. We note that this accounts for the energy costs alone, and neglects other cost factors such as equipment cost, maintenance, labor, etc. However, the electricity costs alone are so prohibitively expensive that it alone is conclusive evidence that a classical attack on a 128 bit ideal cipher key is not economically possible.

\section*{$m$ to 1 detailed proofs}
\begin{theorem}
There exists a $k$-parallel quantum algorithm $A^{f_S(\cdot)}$ such that $\Pr\left[A^{f_S(\cdot)}(x)(1^n) \in S\right] > \frac{1}{2}$ in sequential time $O\left(\sqrt{\frac{N}{kM}}\right)$ and makes $O\left(\sqrt{\frac{kN}{M}}\right)$ oracle queries, where probability is taken over selection of a random subset $S \subseteq \{0,1\}^n$ of size $M$ as well as the randomness of $A^{f_S(\cdot)}$.
\end{theorem}
\begin{proof} WLOG we assume that $M=2^m$ is a power of two to simplify exposition.
Let $A^{f_S(\cdot)}(1^n)$ do the following:
\begin{enumerate}
\item Partition the search space into $m$ blocks $B_{0^m},\ldots, B_{1^m}$ where we have $B_x = \{ xy ~:~y \in \{0,1\}^{n-m}\}$ for each $x \in \{0,1\}^{m}$.
\item Select a block uniformly random $B_x$ for $x \in \{0,1\}^{m}$. 
\item Run a modified $k$-Parallel Grover's algorithm on the block $B_x$.
\end{enumerate}
Straightforward balls and bins analysis tells us that $\Pr\left[ \left|B_x \cap S \right| \geq 1\right] \geq 1-\frac{1}{e}$. Boyer et al~\cite{boyer1998tight} adapted Grover's algorithm to handle the case where there are an unknown number of solutions $t$. Their algorithm runs in sequential time $O\left(\sqrt{\frac{|B_x|}{t}}\right)$. Since, $\sqrt{\frac{|B_x|}{t}} = O\left( \sqrt{|B_x|} \right)$ and $\sqrt{|B_x|} = \sqrt{N/M}$ the running time would be at most $O\left(\sqrt{\frac{N}{M}}\right)$. If the attacker is $k$-parallel we can use the standard trick of further dividing $B_x$ into $k$ blocks $B_{x,1},\ldots, B_{x,k}$ of equal size and running an independent search on each of these blocks. Each of these searches requires sequential time $O\left( \sqrt{|B_{x,i}|} \right) =O\left(\sqrt{\frac{N}{Mk}}\right)$ with $O\left(\sqrt{\frac{N}{Mk}}\right)$ queries to $f_S(\cdot)$ total number of oracle queries would be $O\left(\sqrt{\frac{N}{Mk}}\right)$. 

We remark that if $\left| B_x \cap S \right| \geq 1$ then the search will succeed with high probability. Thus, we have $\Pr\left[A^{f_S(\cdot)}(x)(1^n) \in S\right] \geq \frac{1}{2}$ as required.
\end{proof}

\begin{theorem}
Given any constant $c \in (0,1]$ there is no $k$-parallel quantum algorithm $A^{f_S}$ running in sequential time $o\left(\sqrt{\frac{N}{Mk}}\right)$ and making at most $o\left(\sqrt{\frac{Nk}{M}}\right)$ oracle queries that can find an element $x \in S$ with probability $\Pr\left[A^{f_S}(1^n) \in S\right] > c$, where probability is taken over selection of a random subset $S \subseteq \{0,1\}^n$ of size $m$ as well as the randomness of $A^{f_S(\cdot)}$.
\end{theorem}
\begin{proof} (sketch) We assume that $M=2^m$ is a power of two to simplify presentation. We first note that the problem of selecting a susbset $S$ of size $M$ is equivalent to randomly partitioning the search space $\{0,1\}^N$ into $M$ blocks $B_{0^m},\ldots, B_{1^m}$ of size $2^{n-m}= N/M$ and then constructing $S$ by randomly selecting one element from each block $B_x$. If we offer to reveal the partition $B_{0^m},\ldots, B_{1^m}$ this can only help the attacker. Thus, without loss of generality we can assume that $S$ is constructed by selecting one random element from each of the sets $B_x = \{ yx ~:~y \in \{0,1\}^{n-m}\}$ for each $x \in \{0,1\}^{m}$.

We will argue by contradiction. In particular, we show that if such a $k$-parallel quantum algorithm  $A^{f_S}$ exists such that (1) $\Pr\left[A^{f_S}(1^n) \in S\right] > c$, (2) $A^{f_S}$ runs in sequential time $o\left(\sqrt{\frac{N}{Mk}}\right)$ and (3) $A^{f_S}$  makes at most $o\left(\sqrt{\frac{Nk}{M}}\right)$ oracle queries then we can devise a new $k$-parallel quantum algorithm $A'$ to solve the regular quantum search problem over the search space $\{0,1\}^{n-m}$  such that $A'$ runs in sequential time $o\left( \sqrt{\frac{N'}{k}}\right)$ and makes at most $o\left( \sqrt{k N'}\right)$ queries contradicting a result of Zalka~\cite{zalka1999grover}. 

Given an indicator function $f_x: \{0,1\}^{n-m} \rightarrow \{0,1\}$ such that $f_x(x)=1$ and $f_x(y) = 0$ for all $y \neq x$ the quantum search problem is to find the secret value $x$ given oracle access to $f_x(\cdot)$. Our algorithm $A'^{f_x}$ will select random values $y_z \in \{0,1\}^{n}$ for each $z \in \{0,1\}^m$ subject to the constraint that for any $z \neq z'$ the last $m$ bits of $y_z$ and $y_{z'}$ are distinct. We can implicitly define the set $S = \{ (xz) \oplus y_z~:~ z \in \{0,1\}^m\}$. The set $S$ cannot be constructed explicitly, but $A'$ can simulate the oracle $f_S(\cdot)$ using two queries to the oracle $f_x(\cdot)$. In particular, given an input $w \in \{0,1\}^n$ for $f_S(\cdot)$ there are at most two values $z \in \{0,1\}^m$ such that $w_z = w \oplus y_z = x_w z$ for some string $x_w \in \{0,1\}^{n-m}$ and then let $x_z$ denote the first $n-m$ bits of $w_z$. We remark that $w \in S$ if and only if we can find $z,w$ such that $w_z = w \oplus y_z = x_wz$ and $f_x(x_w)=1$. $A'$ will now simulate   $A^{f_S}$ to recover $w \in S$ with probability at least $c > 0$. The sequential running time of $A'$ will still be $o\left(\sqrt{\frac{N}{Mk}}\right) = o\left(\sqrt{{N'}{k}} \right)$ and the total number of queries will be $q_{A'} =2 * q_A = o\left( \sqrt{N'k}\right)$. Given $w \in S$ we can find the unique value $z \in \{0,1\}^m$ such that $w_z = w \oplus y_z = x_w z$ for some string $x_w \in \{0,1\}^{n-m}$ and recover $x$ from the first $n-m$ bits of $w_z$. 
\end{proof}

\section*{Detailed profit maximization derivation}\label{apdx:derive}
Here we seek to maximize: 
\begin{align*}
P(T,v_0) &= R_{\delta,T}(T,v_0) - C(T) \\
&= v \delta^T - \frac{C_{CCY} \pi^2 N d^2}{16 T s^2}
\end{align*}
We compress via $\Lambda = \frac{C_{CCY} \pi^2 N d^2}{16 s^2}$. Profit can be maximized as:
\begin{align*}
P'(T,v_0) &= \frac{d}{dT_y}\left( v \delta^{T} - \frac{\Lambda}{T_y}\right)\\
&= v \delta^{T_y} \ln \delta + \frac{\Lambda}{T^2}.\\
0 &= v \delta^{T} \ln \delta + \frac{\Lambda}{T^2}\\
\frac{\Lambda}{T^2} &= v\delta^{T} \ln \delta^{-1}\\
\delta^{T} T^2 &= \frac{\Lambda}{v \ln \delta^{-1}}\\
\end{align*}

\end{document}